\documentclass{article}
\usepackage{ijcai21}
\usepackage{amssymb}
\usepackage{times}
\usepackage{soul}
\usepackage{url}
\usepackage[hidelinks]{hyperref}
\usepackage[utf8]{inputenc}
\usepackage[small]{caption}
\usepackage{graphicx}
\usepackage{amsmath}
\usepackage{amsthm}
\usepackage{booktabs}
\usepackage{multirow}
\usepackage[ruled]{algorithm2e}
\usepackage{makecell} 

\newtheorem{theorem}{Theorem}
\usepackage{marvosym}
\usepackage{stfloats}

\usepackage{epstopdf}
\usepackage{enumerate}
\usepackage{color}

\usepackage{flushend}

\newtheorem{assumption}{Assumption}
\newtheorem{lemma}{Lemma}

\title{Efficient Byzantine-Resilient Stochastic Gradient Descent}

\author{
    Kaiyun Li$^{1,2}$
	\and
	Xiaojun Chen$^{1,}$\textsuperscript{\Letter} \and
	Ye Dong$^{1,2}$
	\and 
	Peng Zhang$^{3}$
	\and 
	Dakui Wang$^{1}$
	\and 
	Shuai Zeng$^{1}$
	\affiliations
	$^1$Institute of Information Engineering, Chinese Academy of Sciencess, Beijing, China.\\
	$^2$ School of Cyber Security, University of Chinese Academy of Sciencess, Beijing, China.\\
	$^3$ Guangzhou University, Guangzhou, China.
	\emails
	\{likaiyun, chenxiaojun,dongye,wangdakui,zengshuai\}@iie.ac.cn\\
	p.zhang@gzhu.edu.cn
}

\begin{document}

\maketitle
\begin{abstract}	
\textit{Distributed Learning} often suffers from  Byzantine failures, and there have been a number of works studying the problem of distributed stochastic optimization under Byzantine failures, where only a portion of workers, instead of all the workers in a distributed learning system, compute stochastic gradients at each iteration. These methods, albeit workable under Byzantine failures, have the shortcomings of either a sub-optimal convergence rate or high computation cost. To this end, we propose a new Byzantine-resilient stochastic gradient descent algorithm (BrSGD for short) which is provably robust against Byzantine failures. BrSGD obtains the optimal statistical performance and efficient computation simultaneously. In particular, BrSGD can achieve an order-optimal statistical error rate for strongly convex loss functions. The computation complexity of BrSGD is $O(md)$, where $d$ is the model dimension and $m$ is the number of machines. Experimental results show that BrSGD can obtain  competitive results compared with non-Byzantine machines in terms of effectiveness and convergence.
\end{abstract}
	
\section{Introduction}
\label{Introduction}
In the fields of recommendation systems, natural language processing, and computer vision, it is often the case that we need to build complex models from large-scale  datasets. As the available data for training models continuously grow, it is urgent to use distributed learning to digest these large-scale data. In a  distributed learning system, robustness and security issues have become a major concern. In particular, individual computing units, i.e., \textit{worker machines}, may exhibit abnormal behavior due to data corruption, hardware/software malfunction, communication delay, etc. This abnormal (even worse adversarial) behavior is typically modeled as \textit{Byzantine failure} \cite{lamport2019byzantine}. It is well-known that a single Byzantine-faulty machine can arbitrarily skew standard distributed learning algorithms based on naive gradients aggregation, e.g., an average of the gradients collected from all the workers~\cite{blanchard2017machine}. Moreover, Byzantine failure is exacerbated in Federated Learning (FL), where training data resides at autonomous worker machines, and a central server facilitates the learning process~\cite{konevcny2015federated,mcmahan2017communication,kairouz2019advances}. The behavior of a worker machine under Byzantine failure is often unpredictable, and even becomes susceptible to malicious and coordinated attacks~\cite{aono2017privacy,zhao2020idlg}.  Therefore, it is increasingly important to develop robust distributed algorithms in the adversarial setting. 

Existing works on  ~\cite{blanchard2017machine,feng2014distributed,chen2017distributed,xie2018generalized,yin2018byzantine} suffer from  the shortcomings of either a sub-optimal statistical guarantees or high computation cost. They are inapplicable to distributed settings where. As a result, we wish to develop distributed learning algorithms that can achieve two objectives simultaneously. The first objective is \textbf{ Byzantine-resilience} which refers to achieving the best performance even though a relatively large fraction of the workers are Byzantine. The second objective is \textbf{high computation efficiency} which refers to preserving as much as possible the run-time speedup by using distributing computation across multiple workers.  

To deal with Byzantine failures, an intuitive solution is to use coordinate-wise median \cite{yin2018byzantine} and its variants~\cite{chen2017distributed,xie2018generalized}. However, these methods independently consider each single exception dimension, which requires heavy computation cost when there are a large number of dimensions. For example, to compute the median value of a single exception dimension, the easiest way is to apply a sorting algorithm to the dimension, which yields a time complexity of $O(dm\log m)$. This time complexity is unaffordable when $d$ is very large.

In this paper, we define a new metric for the gradient vector to evaluate whether it is a Byzantine worker. Our key insight is to directly use the $\ell_1$-norm  constraint to eliminate abnormal gradients, which can greatly simplify the calculation. Specifically, let $\boldsymbol{g}^i$ be a stochastic gradient computed by a worker $i\in [m]$, then $\boldsymbol{g}^i$ shall  concentrate around $\boldsymbol{g}^{med}$, where $\boldsymbol{g}^{med}$ is the median value of $\boldsymbol{g}^i,i\in [m]$. In other words, if $\|\boldsymbol{g}^i - \boldsymbol{g}^{med}\|_1 > 2\mathfrak{T}$, where $\mathfrak{T}>0$ is a hyper-parameter, then we can declare worker $i$ as Byzantine.

\begin{figure}[t]
    \centering
    \includegraphics[width=\linewidth]{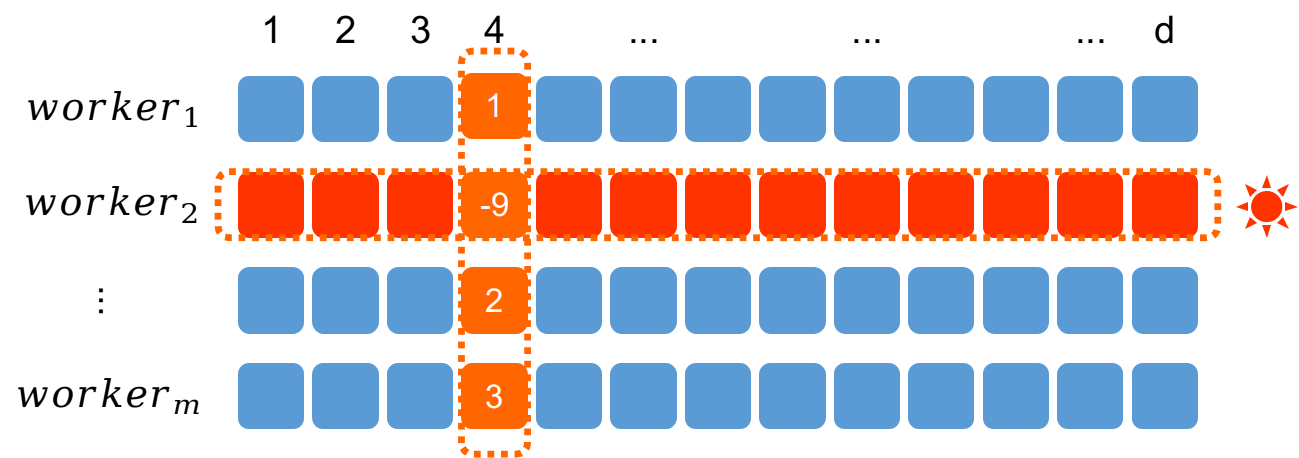}
    \caption{The gradients are collected and combined into a matrix. Taking a column of the matrix, the values uploaded by the normal workers in this column has little difference.}
    \label{idea}
\end{figure}

A challenge to the $\ell_1$-norm method is that some  Byzantine workers may hide without violating the constraint. Then, we further define a score for each gradient. As shown in Fig.\ref{idea},  the master machine  needs to collect the gradients that are flattened into a $d$-dimensional vector from all the workers, and then align these vectors into a new matrix by columns. As shown in Fig.\ref{idea2}, when taking one column of the matrix, we can observe that a subset with a smaller number of elements must deviate heavier from the global mean value than a larger number of elements. Based on the observations that the \textit{honest majority} and the values uploaded by the normal workers are with little difference, a subset with a smaller number of elements is more likely to contain Byzantine workers. Therefore, we can judge the outliers according to each dimension of the gradient and count the number of outliers as a score. We keep the $\beta$-fraction gradients with the highest scores. 

Noted that Byzantine workers may be falsely labeled as good ones. Fortunately, the convergence rates of our method are not impacted significantly. Especially, the subroutine for calculating the score has only one round of comparison and one round of averaging. Thus, it is computationally efficient, almost the same as a naive method implemented through averaging~\cite{polyak1992acceleration}. The  contributions of the paper can be  summarized as follows:
\begin{itemize}
    \item We present a new robust aggregation rule for distributed synchronous Stochastic Gradient Descent algorithm BrSGD in an adversarial setting. The new algorithm can handle Byzantine resilience. Moreover, the computation complexity of the proposed algorithm is $O(md)$.
    \item We theoretically analyze the statistical error rates of our method on strongly convex loss functions. In particular, our algorithm can achieve an order-optimal statistical error rate for strongly convex losses. 
    \item We also demonstrate the convergence of the proposed algorithm by conducting empirical studies. The experimental results on four types of Byzantine attacks match the results of none-Byzantine machines in terms of effectiveness and convergence.
\end{itemize}

The rest of this paper is organized as follows. Related work is summarized in Section \ref{Related_Work}. In Section \ref{Problem_Setup}, we formulate the problem of this work. Section \ref{Algorithms} presents our method and analyzes the statistical error rates. In Section \ref{Experiments}, we present the empirical evaluations, and we conclude this work in Section \ref{conclusion}.

\begin{figure}[t]
    \centering
    \includegraphics[width=\linewidth]{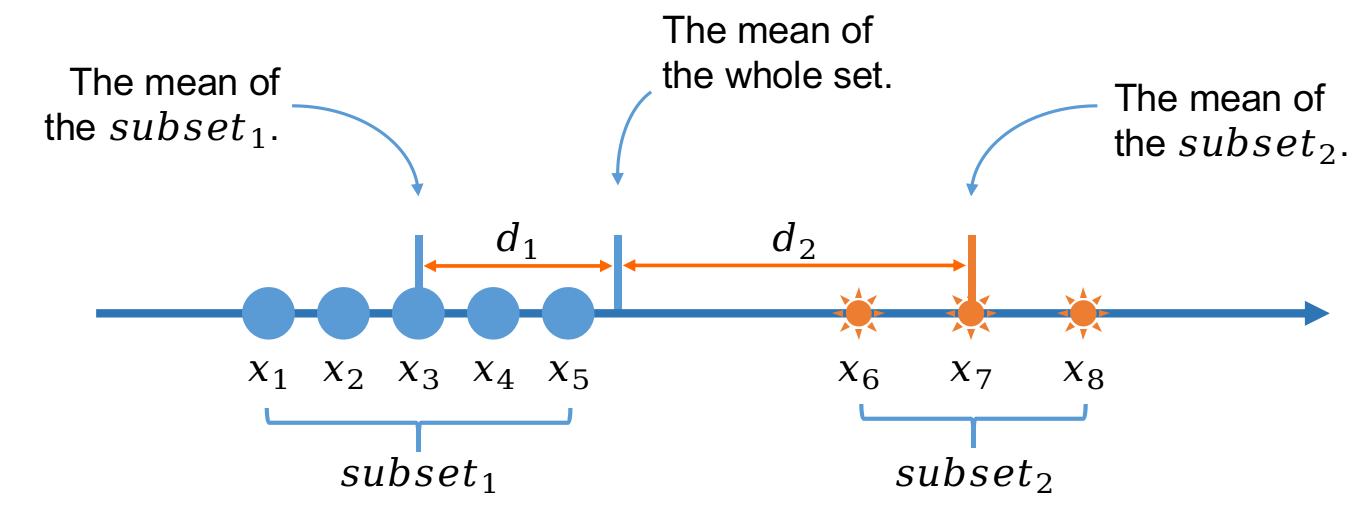}
    \caption{The goal of Byzantine workers ($\mathrm{subset}_2$) is to make the mean value of the set $\{x_1,\dots,x_8\}$ deviating from the \textit{real} mean of the $\mathrm{subset}_1$ as much as possible. However, based on the \textit{honest majority} assumption, the number of elements in $\mathrm{subset}_1$ is greater than $\mathrm{subset}_2$, i.e., $d_1<d_2$. Thus, we can infer whether the Byzantine worker is included according to the number of elements in the two subsets divided by the mean value.}
    \label{idea2}
\end{figure}

\section{Related Work} 
\label{Related_Work}

Recently there have been a large number of works stuyding the problem of  Byzantine stochastic optimization \cite{blanchard2017machine,feng2014distributed,su2016fault,chen2017distributed,alistarh2018byzantine,su2019defending,xie2018generalized,yin2018byzantine,su2018securing,chen2018draco}. Here, we compare existing Byzantine-robust distributed learning algorithms that are most relevant to our work and summarize the comparison in this section.

According to different ideas, existing algorithms can be roughly divided into the following three categories, i.e.,  \textit{Median-based} method,  \textit{Anomaly-detection-based} method, and  \textit{Redundant-gradients-based} method.

The median-based method is used in the context of Byzantine-tolerant distributed learning in papers \cite{blanchard2017machine,feng2014distributed}. The median-based method can be further sub-divided into \textit{Geometric Median}  \cite{chen2017distributed} and \textit{Marginal Median} \cite{xie2018generalized,yin2018byzantine}. For geometric median, the $(1+\epsilon)$-approximate geometric median can be computed in $O(dm\log^3\frac{1}{\epsilon})$ \cite{cohen2016geometric}. An algorithm called the \textit{selection algorithm} \cite{blum1973time} with average time complexity $O(m)$ ($O(m^2)$ is the worst case) to obtain median values. Anomaly detection is performed by running the $T$ iterations \cite{alistarh2018byzantine,allenzhu2021byzantineresilient}. Although computation complexity is considerable for anomaly detection methods, each gradient needs to be checked, and inner product operations are also involved. These operations are very time-consuming in practical.  In these algorithms, the most time-consuming is via redundant gradients methods (ideas from coding theory) \cite{chen2018draco}.

The closest literature is the work of Yin et al. \cite{yin2018byzantine}. They consider a similar Byzantine model, but for \textit{gradient descent} (GD). Reference \cite{blanchard2017machine} proposes a general Byzantine-resilient gradient aggregation rule called \textit{Krum}. This rule has expensive complexity $O(m^2(d+\log m))$ and does not provide a characterization of the statistical errors.  Recent work \cite{chen2017distributed} only applies in the strongly convex setting and is sub-optimal convergence rate. In contrast, our algorithm is efficient ($O(md)$) in computational complexity and achieves order-optimal statistical error rates for strongly convex losses.

\section{Preliminaries}
\label{Problem_Setup}
First of all, we introduce the used notations. A vector is denoted as $\boldsymbol{w}$, and $\boldsymbol{w}_i$ denotes the $i$-th element of $\boldsymbol{w}$. Matrices are denoted as $\boldsymbol{M}$. We let $\boldsymbol{M}_{i,j}$ denote its entry at location $(i,j)$, $\boldsymbol{M}_{i,\cdot}$ denote its $i$-th row, and $\boldsymbol{M}_{\cdot,j}$ denote its $j$-th column. $[N]$ denotes $\{1,2,\dots,N\}$. For any differentiable function $f$, we denote its partial derivative for the $i$-th argument by $\partial_i f$. We index the set of worker machines by $[m]$, and denote the set of Byzantine machines by $\mathcal{B}\in [m]$ (thus $|\mathcal{B}|=\alpha m$).

We consider a distributed computation model with one master machine and $m$ worker machines and that an $\alpha$-fraction of the workers may be Byzantine (where $\alpha < 1/2$). Each worker machine stores $n$ samples, each of which is sampled independently from an unknown distribution $\mathcal{D}$. $\ell_i (\boldsymbol{w};\boldsymbol{x})$ be a loss function of $i$-th worker machine associated with parameter vector $\boldsymbol{w} \in \mathcal{W}\in \mathbb{R}^d$ and the sample $\boldsymbol{x}$, where $\mathcal{W}$ is the parameter space. Denoted the $j$-th data on the $i$-th worker machine by $\boldsymbol{x}^{i,j}$, and $F_i(\boldsymbol{w}):= \frac{1}{n}\sum_{j=1}^n \ell_i(\boldsymbol{w}, \boldsymbol{x}^{i,j})$ denote the empirical risk function for the $i$-th worker. Our goal is to learn a model defined by the parameter that minimizes the population loss:

\begin{equation}
    \boldsymbol{w}^* = \arg \min_{\boldsymbol{w}\in \mathbb{R}^d} F(\boldsymbol{w}),
\end{equation}
where 
\begin{equation}
    F(\boldsymbol{w}):=\mathbb{E}[F_i (\boldsymbol{w})]. 
\end{equation}
The parameter space $\mathcal{W}$ is compact with diameter $D$, i.e., $\|\boldsymbol{w} - \boldsymbol{w}^{'}\|_2 \le D,\forall \boldsymbol{w}, \boldsymbol{w}^{'}\in \mathcal{W}$.
\begin{algorithm}[t]
\SetAlgoLined
\textbf{Input}:initialize parameters $\boldsymbol{w}^0\in \mathcal{W}$, step-size $\eta$, hyperparameters $0 < \beta < \frac{1}{2}$, thresholds $\mathfrak{T}>0$, and iterations $T$\;

\textit{Master machine}: send $\boldsymbol{w}^0$ to all the worker machines\;
    \For{$t = 0,1,\dots,T-1$}{
        \For{all $i\in [m]$}{ 
            \textit{Worker machine $i$}: compute local gradient\;
            \begin{equation*}
                \boldsymbol{g}^i(\boldsymbol{w}^t) = \begin{cases}
                \nabla \ell_i(\boldsymbol{w}^t), \quad & i\in [m]\backslash \mathcal{B},\\
                \ast, & i \in \mathcal{B}.
            \end{cases}
            \end{equation*}
            upload $\boldsymbol{g}^i(\boldsymbol{w}^t)$ to master machine\;
        }
        \textit{Master machine}: compute aggregate gradient
        \begin{equation*}
            \boldsymbol{g}(\boldsymbol{w}^t) \gets \mathcal{A}_{\beta, \mathfrak{T}}(\{\boldsymbol{g}^i(\boldsymbol{w}^t):i\in [m]\}),
        \end{equation*}
        and then send $\boldsymbol{g}(\boldsymbol{w}^t)$ to all worker machines\;
    
        \textit{Worker machine}: update model parameter
        \begin{equation*}
            \boldsymbol{w}^{t+1}\gets \boldsymbol{w}^t - \eta \boldsymbol{g}(\boldsymbol{w}^t)
        \end{equation*}
        }
 \caption{Byzantine-resilient SGD (BrSGD)}
 \label{alg_overview}
\end{algorithm}

\section{Our Method}
\label{Algorithms}
This section presents our Byzantine-Resilient stochastic gradient descent algorithm (named BrSGD) and briefly summarizes our convergence results on its performance.

\subsection{BrSGD Algorithm}

The master machine broadcasts the initialized model, and all the workers receive and accept it as the local initialized model. Then, the master and the workers update the model iteratively. At each iteration, the normal workers compute the gradients of their local loss functions and then upload to the master. The Byzantine workers may send any messages. Next, the master performs a robust aggregator to compute gradients for model updating and send the gradient to all workers. After receiving the gradient from the master, the workers update its parameters in the way of gradient descent and moves into the next iteration until the whole algorithm is completed. The formulation is illustrated in Algorithm \ref{alg_overview}.


In the following context, we focus on the aggregation method construction. Intuitively, the aggregation rule should return a vector $\boldsymbol{g}$ that is not too far from the \textit{real} gradient $\nabla F(\boldsymbol{w})$. More precisely, $\boldsymbol{g}$ should approximate the steepest direction of the loss function being optimized. This is expressed as upper bound $\|\boldsymbol{g}^i - \boldsymbol{g}^{med}\|_1 \le 2\mathfrak{T}$ on worker's gradient vector $\boldsymbol{g}^i$ (\textit{Constraint 1}). If the gradient satisfies Constraint 1, we add it to the candidates $\mathcal{C}_1$. To further eliminate the impacts of potential anomaly values in candidates, we set a score to each $\boldsymbol{g}^i$ and keep the $\beta$-fraction gradients with larger scores (\textit{Constraint 2}). Concretely, the master collect the gradients that are flattened into a $d$-dimensional vector from all workers and then align them into a new matrix $\boldsymbol{G}$ by column. Then, we utilize the observation (cf. Section \ref{Introduction}) to mark the anomaly values of $\boldsymbol{G}$ column-wise. We divide one column of data into two subsets using their mean and set the subset with a large number of elements to 1, otherwise 0. Finally, a scoring matrix $\boldsymbol{M}$ composed of 0 and 1 will be generated, and the result of adding the rows will be used as a score for the corresponding gradient. For each worker $i$, we define the score $\boldsymbol{s}_i = \sum_{j\in [d]}\boldsymbol{M}_{i,j},i \in [m]$. For the details, please see Algorithm \ref{aggregater}.

\subsection{Convergence Analysis}

\begin{algorithm}[t]
\SetAlgoLined
\textbf{Inputs}: gradient set $\{\boldsymbol{g}^i(\boldsymbol{w}^t), i\in [m]\}$, hyperparameters $0 < \alpha < \beta \le \frac{1}{2}$, and thresholds $\mathfrak{T}>0$\;
\textbf{Return}: aggregation gradient $\boldsymbol{g}(\boldsymbol{w}^t)$\;
\vspace{0.3cm}
    
    $\boldsymbol{a} = 0,\boldsymbol{s} = 0, \boldsymbol{n} = 0, \boldsymbol{M} = \boldsymbol{0};$\quad \# Initialize the mean vector $\boldsymbol{a}$, score vector $\boldsymbol{s}$, $\ell_1$-norm vector $\boldsymbol{n}$, and score matrix $\boldsymbol{M}$\; 
    \For{$c = 0,1,\dots,d$}{
        \For{$r = 0,1,\dots,m$}{ 
               $\boldsymbol{a}_c = \boldsymbol{a}_c + \boldsymbol{a}_r$\;
        }
        $\boldsymbol{a}_c = \boldsymbol{a}_c / m;$\quad \# calculate the average value of the arranged gradient matrix by column\;
        }
    \For{$c = 0,1,\dots,d$}{
    $counter = 0$\;
        \For{$r = 0,1,\dots,m$}{ 
        \If{$\boldsymbol{g}_c^r(\boldsymbol{w}^t) \ge   \boldsymbol{a}_c$}{
            $\boldsymbol{M}_{r,c} = 1$\;
            $counter = counter + 1$\;
        }
        \Else{
            $\boldsymbol{M}_{r,c} = 0$\;
        }}
        \If{$counter < m/2$}{
            $\boldsymbol{M}_{\cdot,d} = \lnot \boldsymbol{M}_{\cdot,d};$\quad \# obtain the score matrix based on observation 1\;
        }
        $\boldsymbol{s} = \boldsymbol{s} + \boldsymbol{M}_{\cdot,d}^T;$\quad \# calculate the score for each gradient\;
        }
        $\boldsymbol{g}_{med} = \mathrm{median}\{\boldsymbol{g}^{i}(\boldsymbol{w}^t):i\in [m]\}$\;
        $\mathcal{C}_1= \{i:\|\boldsymbol{g}^i(\boldsymbol{w}^t) - \boldsymbol{g}_{med}\|_1 \le 2\mathfrak{T}\};$\ \# \textit{Constraint 1}\;
        $\mathcal{C}_2 = \{i: \mathrm{Top}_{max}^\beta(\boldsymbol{s}_i),i\in [m]\};$\ \# \textit{Constraint 2}\;       
        $\boldsymbol{g}(\boldsymbol{w}^t) = \textrm{mean}(\boldsymbol{g}^i(\boldsymbol{w}^t):i\in \mathcal{C}_1 \cap \mathcal{C}_2)$\;
 \caption{Aggregator $\mathcal{A}_{\beta, \mathfrak{T}}(\{\boldsymbol{g}^i(\boldsymbol{w}^t):i\in [m]\})$}
 \label{aggregater}
\end{algorithm}

In this part, we provide statistical guarantees on the error rates (defined as the distance between $\boldsymbol{w}^T$ and the optimal solution $\boldsymbol{w}^{*}$) of the algorithms. Throughout we assume that the loss function $F(\boldsymbol{w})$ is   smooth.

\begin{assumption}
\label{ass_1}
In each iteration $t$, each normal worker machine $i\in [m]$ gives back a vector $\nabla \ell_i(\boldsymbol{w}^t)\in \mathbb{R}^d$ satisfying $\|\nabla \ell_i(\boldsymbol{w}^t) - \nabla F(\boldsymbol{w}^t)\|_1 \le \mathcal{V}$.
\end{assumption}

One can instead assume $Pr[\|\nabla \ell_i(\boldsymbol{w}^t) - \nabla F(\boldsymbol{w}^t) \|_1\ge t]\le 2\exp (-t^2/2\mathcal{V}^2)$ and the results of this paper continue to hold up to logarithmic factors. But we assumes $F(\boldsymbol{w}_t)$ is Lipschtz continuous, which implies $\|\nabla \ell_i(\boldsymbol{w}^t) - \nabla F(\boldsymbol{w}^t)\|_1$ is bounded. To present the simplest theory, we assumed it is bound with probability 1. 

Firstly, we introduces Lemma \ref{lemma1}, which is used to bound the output of algorithm \ref{aggregater}.
\begin{lemma}
\label{lemma1}
In our algorithm, suppose that the one dimensional samples on all the normal worker machines are independent and identically distributed (i.i.d). Satisfy $v$-sub-exponential with mean $\mu$. For any $t\ge 0$ and $s\ge 0$, let $|\frac{1}{|\mathcal{C}_1\cap \mathcal{C}_2|}\sum_{i\in [m]\backslash \mathcal{B}}g^i - \mu|\le t$ and $\max_{i\in [m]\backslash \mathcal{B}}|g^i - \mu| \le s$. If we set $\mathfrak{T} = s \le \mathcal{V}$, then we have
\begin{equation}
    |\mathcal{A}_{\beta,\mathfrak{T}}(\{g^i:i\in [m]\}) - \mu| \le t+3\beta s
\end{equation}
\end{lemma}

\begin{proof}
To simplify notation, we define $\mathcal{M}=[m]\backslash\mathcal{B}$ as the set of all normal worker machines, $\mathcal{C} = \{\mathcal{C}_1\cap \mathcal{C}_2\} \subseteq [m]$ as the set of candidates, and $\mathcal{R}\subseteq [m]$ as the set of all removed machines. The estimator simply computes 
\begin{equation}
   \mathcal{A}_{\beta,\mathfrak{T}}(\{g^i:i\in [m]\}) = \frac{1}{|\mathcal{C}|}\sum_{i\in \mathcal{C}} g^i
\end{equation}
Condition 1 and Assumption \ref{ass_1} ensures that the boundary of each dimension of the gradient in the candidate set does not exceed $2s$. Then, We have 
\begin{equation}
\begin{aligned}
    &|\mathcal{A}_{\beta,\mathfrak{T}}(\{g^i:i\in [m]\}) - \mu| = \left| \frac{1}{|\mathcal{C}|}\sum_{i\in \mathcal{C}} g^i - \mu \right|\\
    &=\frac{1}{|\mathcal{C}|}\left|\sum_{i\in \mathcal{M}} (g^i - \mu) - \sum_{i\in \mathcal{M} \cap \mathcal{R}} (g^i - \mu)+ \sum_{i\in \mathcal{B}\cup \mathcal{C}} (g^i - \mu)\right|\\
    &\le \frac{1}{|\mathcal{C}|}\left((1-\alpha)mt + \beta ms + 2\alpha ms \right)\\
    &\le t+3\beta s
\end{aligned}
\end{equation}
where the last inequality is by $\beta \ge \alpha$. We directly obtain the desired result. 
\end{proof}

\begin{assumption}[Smoothness of $\ell$ and $F$]
\label{ass_2}
For any $\boldsymbol{x}\in \mathcal{X}$, the partial derivative of $\ell(\cdot; \boldsymbol{x})$ with respect to the $k$-th coordinate of its first argument, denoted by $\partial_k \ell(\cdot; \boldsymbol{x})$, is $L_k$-Lipschitz continuity for each $k\in [d]$, and the function $\ell(\cdot; \boldsymbol{x})$ is $G$-Lipschitz smoothness. Let $L:=\sqrt{\sum_{k=1}^dL_k^2}$. Also assume that the population loss function $F(\cdot)$ is $G_F$-Lipschitz smoothness. It is easy to see that $G_F\le G\le L$.
\end{assumption}

\subsubsection{Strongly Convex Losses} We consider the case where the population loss function $F(\cdot)$ is strongly convex.

\begin{theorem}
\label{theorem_2}
suppose that Assumption \ref{ass_2} hold, $F(\cdot)$ is $\lambda_F$-strongly convex,  and $\alpha \le \frac{1}{2} - \epsilon$ for some $\epsilon > 0$. In addition, we assume that $\partial_k \ell(\boldsymbol{w};\boldsymbol{x})$ is $v$-sub-exponential for any $k\in [d]$ and $\boldsymbol{w}\in \mathcal{W}$. Choose step-size $\eta=1/G_F$. Then, with probability at least $1 - \frac{4d}{(1+nmLD)^d}$, after $T$ parallel iterations, we have 
\begin{equation*}
    \| \boldsymbol{w}^T - \boldsymbol{w}^* \|_2 \le \left(1 - \frac{\lambda_F}{G_F + \lambda_F} \right)^T \| \boldsymbol{w}^0 - \boldsymbol{w}^* \|_2 + \frac{2}{\lambda_F}\triangle
\end{equation*}
where
\begin{equation*}
\label{triangle}
    \triangle = \widetilde{O}\left(\frac{1}{\sqrt{n}}+\frac{1}{\sqrt{nm}} \right)
\end{equation*}
\end{theorem}
Note that this paper \cite{yin2018byzantine} provide a lower bound showing that the error rate of any algorithm is $\widetilde{\Omega}(\frac{\alpha}{\sqrt{n}}+\frac{1}{\sqrt{nm}})$. Therefore, our algorithm achieved order-optimal statistical error rate.  
\begin{proof}
The proof is essentially the same as \cite{yin2018byzantine}. Therefore, we omit the details of the analysis here. Choosing 
\begin{small}
\begin{equation*}
t = v\max \left\{ \frac{8d}{nm}\log(1+nmLD)
 , \sqrt{\frac{8d}{nm}\log(1+nmLD)}  \right\}, 
 \end{equation*}
 \end{small}
 and 
\begin{equation*}
\begin{aligned}
s &= v\max \left\{ \frac{4}{n}(d\log(1+nmLD)+\log m)\right.\\
        &\qquad \qquad \qquad \left. , \sqrt{\frac{4}{n}(d\log(1+nmLD)\log m)} \right\}. 
    \end{aligned}
\end{equation*}
By using the simple algebra, we have
\begin{equation*}
    \triangle = \widetilde{O}\left(\frac{1}{\sqrt{n}}+\frac{1}{\sqrt{nm}} \right)
\end{equation*}
\end{proof}

\section{Experiments}
\label{Experiments}
\begin{table*}
  \centering
  \begin{tabular}{c|llll|lll|lll|lll}
    \toprule
    Algorithms & \multicolumn{4}{c}{Ours}  &    \multicolumn{3}{c}{Median}     &    \multicolumn{3}{c}{Mean}     &    \multicolumn{3}{c}{Krum}   \\
    \cmidrule(r){1-14}

    $\alpha$            & 0     & $10\%$  & $25\%$ & $50\%$ & $10\%$  & $25\%$ & $50\%$ & $10\%$  & $25\%$ & $50\%$ & $10\%$  & $25\%$ & $50\%$\\
    \midrule
    Gaussian            & 84.16     & 83.42 & 84.56 & 83.29  & 83.73 & 85.40 & 86.78 & 9.93 & 9.93 & 9.93    & 82.04 & 79.56 & 79.62\\
    Model Negation      & 85.87     & 83.89 & 84.07 & 82.13  & 83.29 & 84.47 & 82.84 & 10.84 & 10.84 & 10.84 & 80.44 & 81.18 & 79.36\\
    Gradient Scale      & 85.29     & 84.33 & 83.67 & 82.87  & 81.00 & 61.13 & 9.80  & 85.87 & 82.47 & 60.67 & 77.20 & 81.33 & 70.76\\
    Label Inverse       & 85.11     & 82.93 & 82.49 & 83.93  & 84.91 & 86.00 & 82.09 & 72.71 & 74.93 & 80.11 & 81.82 & 78.56 & 68.76\\

    \bottomrule
  \end{tabular}
  \caption{Test results of LeNet on FashionMNIST in terms of accuracy using stochastic gradient descent.}
  \label{tableacc}
\end{table*}

We implemented BrSGD in Python3. Our experiments are executed on Intel$\circledR$ Xeon$\circledR$ CPU E5-2650 v3@ 2.30GHz servers with 64GB RAM. We simulate $m=20$ clients. The model is learned atop Pytorch v1.6.0 equipped with CUDA v10.2 and one 12G memory TITAN Xp GPU. the C-S connection is over WAN with 50Mbps bandwidth and 50ms RTT.

\subsection{Experimental settings}

We conduct experiments to show the Byzantine resilience in the following four classic adversarial settings: 
\begin{itemize}
    \item We consider the attackers that replace some of the gradient vectors with  Gaussian random vectors with zero mean and isotropic covariance matrix with standard deviation 200. We refer to this kind of attack as \textit{Gaussian Attack}. 
    \item For each Byzantine gradient vector, the gradient is replaced by the negative sum of all the correct gradients, scaled by a large constant (1e10 in the experiments). We call this attack \textit{Model Negation Attack}. 
    \item We generate the Byzantine machines in the following way: we replace every training label $y$ on these machines with $9-y$, e.g., 1 is replaced with 8, 7 is replaced with 2, etc., and the Byzantine machines compute gradients on these data. Denote this kind of attack by \textit{Label Shift}. 
    \item Named after \textit{Gradient Scale Attack}, an attacker replaces some of the gradient vectors with scaled by a constant (1e10 in the experiments).
\end{itemize}

We instantiate $m=20$ workers and one master node in experiments, and train a convolutional neural network model (LeNet \cite{lecun1998gradient}) on the FashionMNIST dataset \cite{Xiao2017fashion} using mini-batch stochastic gradient descent to compare the test accuracies in those above four adversarial settings. We compare accuracy and convergence with \textit{Mean} (find the average of all gradients), \textit{Krum}\cite{blanchard2017machine}, and coordinate-wise median\cite{yin2018byzantine}. 
\subsection{Experimental results}
\begin{figure*}[t]
  \centering
  \includegraphics[width=0.9\linewidth]{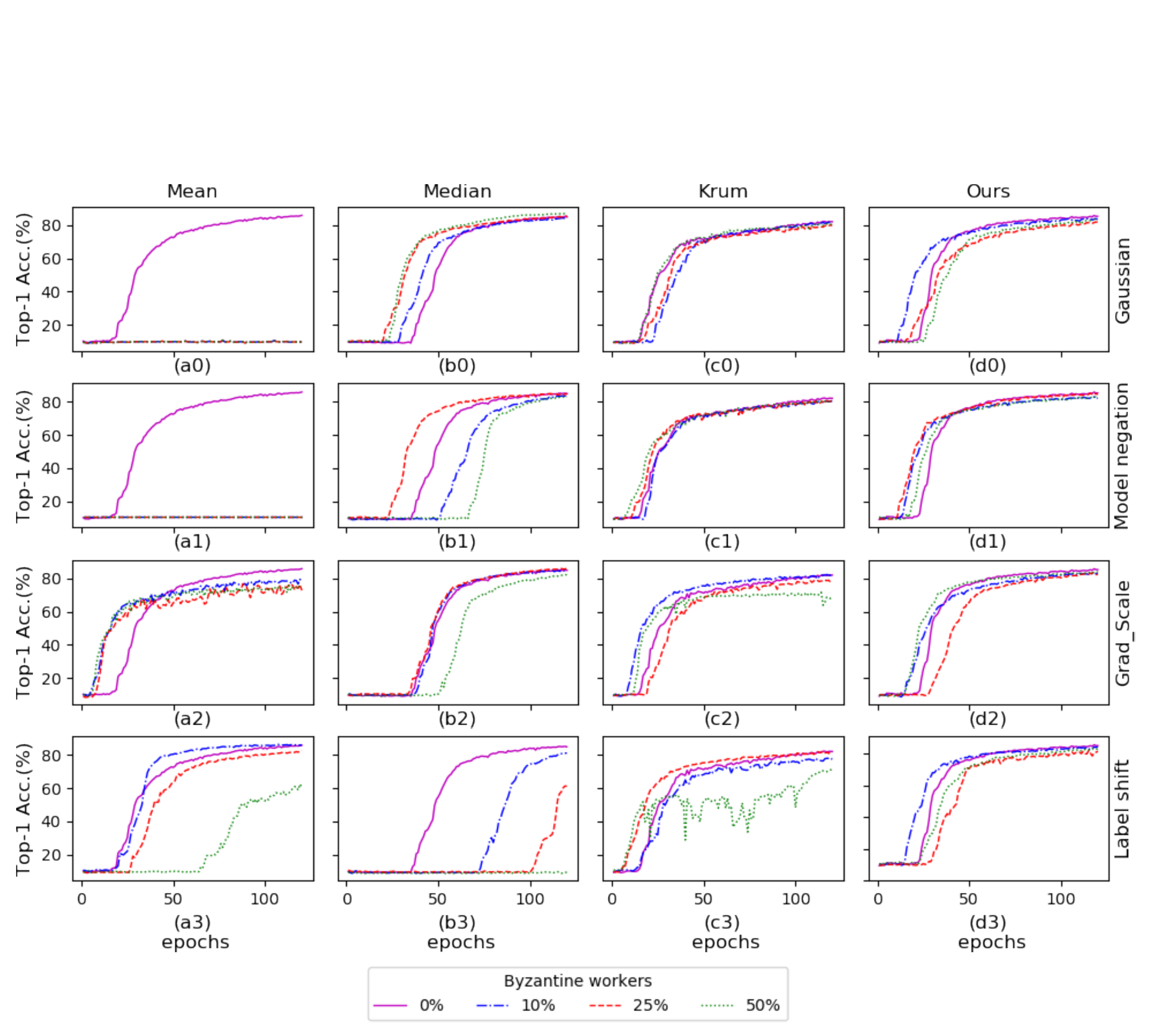}
  \caption{The results \textit{w.r.t.} accuracy vary with the epochs under Byzantine attacks, where we set step-size $\eta=0.03$ for the LeNet on FashionMNIST. Curves correspond to losses, and columns correspond to different aggregation methods. Byzantine workers perform four attacks with 10\%, 25\%, and 50\% attackers, respectively. We set $\beta=1/2$ for our algorithm.}
  \label{pic_acc}
\end{figure*}

As is shown in Figure \ref{pic_acc},  we observe that even with a single attacker performing a \textit{Gaussian Attack} or \textit{Model Negation Attack}, using the \textit{Mean} method directly is devastating (a0,a1). Although the median-based method has high accuracy, the convergence is very slow (b1,b3). \textit{Krum} and our method have good accuracy, but the effect of \textit{Krum} will drop precipitously when the number of Byzantine worker machines is large ($\alpha = 50\%$,c3). In contrast, we can see that our algorithm has the same convergence as the average method without an attack ($\alpha=0$). When the $\alpha=50\%$, the model still achieves high accuracy. We also include the full test accuracy comparison table in Table \ref{tableacc}.
  
In conclusion, the distributed gradient descent algorithm suffers from severe performance loss in adversarial settings. Moreover, show our algorithm can indeed defend against Byzantine failures. Compared with the median-based and \textit{Krum} methods, our algorithm has excellent accuracy and faster convergence for the four attacks.

\section{Conclusion}  
\label{conclusion}
In this paper, we study a new distributed stochastic optimization algorithm in an adversarial setting with the purpose of obtaining the optimal statistical results and computation efficiency simultaneously. Based on the honest-majority assumption, we propose a new stochastic gradient descent algorithm BrSGD. We show that the method can achieve an order-optimal $\widetilde{O}\left(\frac{1}{\sqrt{n}}+\frac{1}{\sqrt{nm}} \right)$ for strongly convex losses, and the computation complexity of our algorithm is $O(md)$. Moreover, we conduct extensive experiments to show that our method outperforms the state-of-the-art methods in terms of effectiveness and convergence.

\bibliographystyle{named}
\bibliography{ijcai21}

\end{document}